\pdfoutput=1

\documentclass[12pt]{article}
\usepackage[utf8]{inputenc}
\usepackage[T1]{fontenc}
\usepackage{fourier}
\usepackage{microtype}
\usepackage{amsmath,amsthm,amssymb}
\newtheorem{proposition}{Proposition}
\usepackage{graphicx,xcolor}
\graphicspath{{./tools/}{.}}
\usepackage[english,francais]{babel}
\usepackage{hyperref}
\usepackage{authblk}

\begin{document}

\title{Estimation de ligne de base de capteurs d'humectation : intégration et minimum locaux à différentes échelles}
\author[1]{Jean-Yves Baudais}
\author[2]{Melen Leclerc}
\author[2]{Christophe Langrume}
\affil[1]{IETR, UMR CNRS 6164, Beaulieu, Rennes, France}
\affil[2]{IGEPP, UMR 1349 INRAE, Domaine de la Motte, Le Rheu, France}
\date{}
\maketitle
\noindent \small jean-yves.baudais@insa-rennes.fr, \{melen.leclerc,christophe.langrume\}@inrae.fr

\renewcommand{\abstractname}{Résumé}
\begin{abstract}  
Les capteurs diélectrique d'humectation sont utilisés en agriculture pour détecter la présence d'eau déposée sur le feuillage et prédire le risque de développement de maladies. Le signal électrique mesuré présente une dérive du niveau de base qui vient biaiser les alertes. Nous proposons une méthode d'estimation de cette ligne de base exploitant la norme L1 et sélectionnant des minimums locaux à une échelle d'observation. Les performances de l'estimateur sont évaluées sur des données simulées et comparées à celles des estimateurs de la littérature.
\end{abstract}

\renewcommand{\abstractname}{Abstract}
\begin{abstract} 
Dielectric wetness sensors are used in agriculture to detect the presence of water on foliage and to predict the risk of disease development. The measured electrical signal has a base level drift that skews the alerts. We propose a method for estimating this baseline using L1 and selecting local minimums at an observation scale. The performance of the estimator is evaluated on simulated data and compared to the literature estimators.
\end{abstract}

\section{Introduction}

L'humectation foliaire, c.-à-d. l'eau présente sur les feuilles, est une donnée importante pour prédire le développement de nombreuses maladies des plantes \cite{MAGA05}. Elle mesurée \emph{in situ} avec des capteurs diélectriques, le plus utilisé étant le capteur LWS \cite{LWS16}. Ce capteur capacitif reproduit grossièrement la forme et les propriétés thermodynamiques, radiatives et hydrophobes d'une feuille. Il mesure la constante diélectrique qui dépend de la quantité d'eau présente sur le capteur \cite{THIE18}. Le signal produit correspond à une succession d'excitations indiquant la présence d'eau sur le capteur (p.~ex. rosée, pluie, gel). Le traitement de ce signal consiste généralement à extraire par seuillage les périodes sèches et humides qui sont ensuite utilisées pour des recherches en épidémiologie végétale ou pour réaliser des alertes, en temps réel, afin de limiter le développement de maladies sur les cultures. Néanmoins, lorsque les capteurs sont laissés pendant de longues périodes en extérieur, une dérive lente du niveau de base du signal apparaît et vient biaiser la détection de l'humectation.

L'estimation et la correction de la ligne de base est une étape de pré-traitement essentielle pour analyser, quantifier et interpréter correctement des mesures physiques, chimiques ou biologiques \cite{BERT06}. De nombreuses méthodes sont proposées pour réaliser cette estimation de manière quasi-automatique une fois l'acquisition réalisée \cite{PEAR77, DIET91}. Elles s'appuient sur l'estimation des pics \cite{DUVA15}, la régression quantile \cite{KOMS11}, la décomposition en ondelettes \cite{QIAN17}, l'ajustement de polynômes \cite{MAZE05} ou des moindres carrés penalisés ou pondérés \cite{ZHAN10}. Ces méthodes nécessitent soit l'ensemble des données, soit l'optimisation d'hyperparamètres pour estimer le niveau de base et corriger le signal. Nous proposons ici une méthode estimant la ligne de base en temps réel et nécessitant une faible profondeur de mémoire. La méthode proposée offre la possibilité d'un traitement \emph{in situ}, c.-à-d. dans un noeud de capteur, afin d'améliorer les outils d'aide à la décision en agriculture pour la lutte contre les maladies et le gel.

Le reste de l'article est organisé de la façon suivante. Le paragraphe~\ref{sec_mm} présente le modèle d'humectation utilisé et la méthode d'estimation de la ligne de base. Le paragraphe~\ref{sec_sr} donne quelques résultats de simulation, quantifie les performances de la méthode proposée et les compare aux performances d'algorithmes de la littérature. Le paragraphe~\ref{sec_conc} conclut l'article.

\section{Modèle et méthode}
\label{sec_mm}

La figure~\ref{fig_data} donne un exemple de 9858 points de mesure d'humectation dans un couvert de pois, obtenus de mars à juin 2014. La préconisation du constructeur des capteurs est de fixer la ligne de base à 274~mV pour une tension d'excitation de 2,5~V, ou d'adapter empiriquement ce seuil à chaque capteur. Aucune des deux méthodes ne donnent de résultats pertinents pour ensuite positionner le seuil d'humectation. Par exemple, les données complètes obtenues avec plusieurs capteurs donnent des valeurs minimales sur les cinq premiers jours allant de 260,5 à 264,8~mV, des valeurs bien en deçà de la ligne de base préconisée.

\begin{figure}
\centering
\includegraphics[width=1\columnwidth]{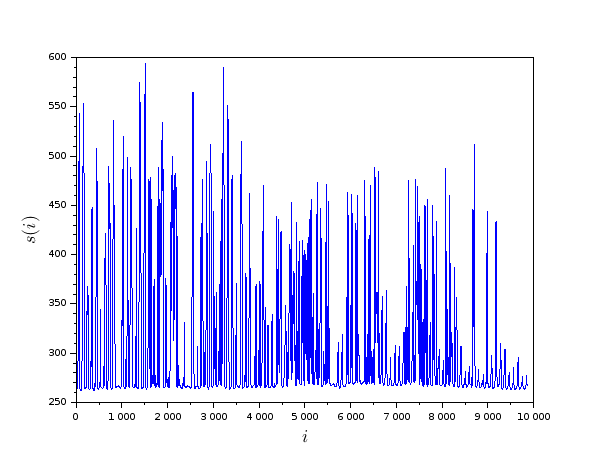}
\caption{Mesure expérimentale d'humectation acquises dans un couvert de pois, amplitude en mV en fonction de l'indice de la mesure.}
\label{fig_data}
\end{figure}

\paragraph{Hypothèse et caractéristique du signal}

L'hypothèse d'additivité des composantes du signal en sortie du capteur conduit à
\begin{align}
s(t)=h(t)+b(t)+n(t)\,,
\end{align}
où $s(t)$ est la réponse du capteur, $h(t)$ est le signal d'humectation à proprement parler, $b(t)$ est la ligne de base et $n(t)$ le bruit électronique ou électromagnétique de l'environnement. Les mesures sont échantillonnées et moyennées sur 15 minutes, d'où un bruit faible et entre 8500 et 10000 échantillons par capteur sur une campagne de mesure. Le modèle discret est alors, $i\in[1,N]\subset\mathbb{N}$
\begin{align}
s(i)=h(i)+b(i)+n(i)\,.
\end{align}
Il n'y pas de fonction ou fonctionnelle à priori modélisant la ligne de base et qui pourrait être exploitée par l'estimateur. Cependant, en plus de l'hypothèse d'additivité, la ligne de base a des amplitudes globalement plus faibles, des variations globalement plus lentes que les amplitudes et variations du signal d'humectation propre, sans pour autant présenter de valeurs simplement séparables. L'estimation de la ligne de base est alors faite sans hypothèse supplémentaire sur $h(i)$ ni sur $b(i)$.

\paragraph{Modèle}

Pour éprouver la méthode d'estimation et évaluer ses performances, un modèle simple est proposé où
\begin{enumerate}
\item $h(i)$ est un mélange de gaussiennes obtenu à partir de la répartition empirique de l'amplitude des pics, de leur largeur et de la distance entre les pics. La figure~\ref{fig_mode} illustre un exemple de comparaisons entre les mesures expérimentales et le modèle ;
\item $b(i)$ est une marche aléatoire de pas $\pm0{,}1$ lissée par moyenne glissante,
\item $n(i)$ est un processus gaussien i.i.d. centré.
\end{enumerate}
Une meilleure modélisation du signal d'humectation serait d'utiliser des gaussiennes généralisées avec différents paramètres de forme permettant de prendre en compte l'asymétrie des lobes ou des fronts plus ou  moins abruptes. Cependant, et dans la mesure où il n'y a pas d'estimation des paramètres du modèle, la forme des fonctions utilisées n'est pas un point clé de notre étude.

\begin{figure}
\centering
\includegraphics[width=1\columnwidth]{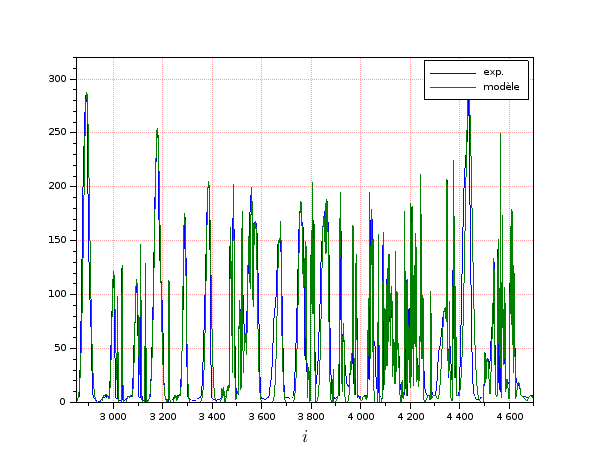}
\caption{Données expérimentales $s(i)$ normalisées et modèle d'humectation $h(i)$, amplitude en mV en fonction de l'indice de la mesure.}
\label{fig_mode}
\end{figure}

\paragraph{Mesure}

Le seuil d'humectation est un paramètre qui dépendant du phénomène biologique à surveiller. Quelle que soit sa valeur, il doit suivre la ligne de base afin de séparer au plus juste les phases sèches des phases humides. Une phase est humide si l'amplitude $s(i)$ du signal est au-dessus de ce seuil $S$, elle est sèche sinon. En appelant $\hat b(i)$ la ligne de base estimée, la comparaison du signal $s(i)$ au seuil théorique $b(i)+S$ et au seuil estimé $\hat b(i)+S$ permet de définir les fausses alarmes (FA), sous entendu vis-à-vis de la phase humide,
\begin{align}\label{eq_fa}
\mathcal{I}_\text{FA}=\big\{i|\hat b(i)+S<s(i)\leq b(i)+S\big\}
\end{align}
et les détections manquées (MD)
\begin{align}\label{eq_md}
\mathcal{I}_\text{MD}=\big\{i|b(i)+S<s(i)\leq \hat b(i)+S\big\}\,.
\end{align}
La durée d'humectation étant une donnée importante d'un point de vue biologique, la distance $d(i)$ à la phase humide est mesurée et définie par
\begin{align}\label{eq_d}
\forall i\in\mathcal{I}_\text{FA}\cup\mathcal{I}_\text{MD},\ d(i)=\min_{\big\{j\big|s(j)>b(j)\wedge s(j)>\hat b(j)\big\}}|i-j|\,.
\end{align}
Le modèle permet de donner des probabilités ou densités de probabilité de FA, MD et $d(i)$, en plus de la classique erreur quadratique moyenne (EQM). Dans le cas des données de terrain, le seuil fixe $S$ et le seuil adaptatif $\hat b(i)+S$ sont utilisés à la place des seuils $\hat b(i)+S$ et $b(i)+S$.

\paragraph{Méthode d'estimation}

L'estimation de la ligne de base est faite à partir de $s(t)$. Elle est d'autant plus fiable au point $t$ que $h(t)$ atteint une valeur nulle, qui est un minimum de la fonction, avec des variations au voisinage de $t$ suffisamment faibles. Ce voisinage correspond alors à une période sèche. Cette observation permet de sélectionner un ensemble de points favorables pour être des points de la ligne de base et ensuite de les interpoler pour reconstruire la ligne de base estimée.

Plutôt que d'utiliser la dérivée, une propriété locale de la fonction sensible au bruit, on minimise une fonction coût $C_T(t_0)$ exploitant l'intégrale à différentes échelles d'observation $T$
\begin{align}
C_T(t_0)=\frac{1}{2T}\int_{t_0-T}^{t_0+T}\phi(s(\tau)-s(t_0))d\tau\,,
\end{align}
où $t_0$ est un minimum local à l'échelle observée
\begin{align}\label{eq_t0}
t_0\ |\ s(t_0)=\min_{\tau\in[-T,T]}s(t_0+\tau)\,.
\end{align}
Les méthodes itératives cherchent à minimiser l'influence des pics via la fonction $\phi$, en exploitant des coûts différents des coûts quadratiques \cite{MAZE05}. Dans notre cas, au contraire, les pics sont utilisés pour maximiser le coût et rejeter ainsi les minimums locaux en période humide, peu fiables pour estimer $b(t)$. Nous avons choisi $\phi(x)=|x|$ qui donne des résultats satisfaisants.

Il est immédiat que l'estimateur ne peut pas suivre les variations de la ligne de base si le signal $h(t)$ est nul pour tout $t$. Dans ce cas, seuls les minimums locaux pourront correctement être estimés. En pratique, ce cas n'existe pas car il y a les phénomènes de rosée nocturne qui conduisent nécessairement à des valeurs de $h(t)$ non nulles.

\begin{proposition}\label{theo_prop}
Si $C_T(t_0)<\epsilon$ alors $s(t_0)$ et un bon candidat pour l'estimation de $b(t_0)$.
\end{proposition}

\begin{proof}
La formulation de la proposition est probabiliste mais il n'y a pas de preuve formelle de l'implication dans la mesure où il est toujours possible de construire des fonctions et des points qui ne vérifient pas la proposition. Cependant, $\forall t_0$~\eqref{eq_t0}
\begin{align}
C_T(t_0)
&\leq\frac{1}{2T}\int|h(t)-h(t_0)|dt+\frac{1}{2T}\int|b(t)-b(t_0)|dt\notag\\&\qquad+\frac{1}{2T}\int|n(t)-n(t_0)|dt\notag\\
&\leq a_h(t_0)+a_b(t_0)+a_n(t_0)\,.
\end{align}
Ainsi,
\begin{itemize}
\item Sur une période sèche, $a_h(t_0)<\epsilon_h(t_0)$ car $h(t)\approx h(t_0)\approx0$. À l'opposé et autour d'un minimum local $s(t_0)$, les variations de $h(t)$ sur une période humide sont plus importantes ;
\item La ligne de base ayant des variations faibles, $b(t)\approx b(t_0)$ sur une fenêtre $2T$ pas trop grande, d'où $a_b(t_0)<\epsilon_b$ ;
\item Enfin, l'intégration permet de réduire l'influence du bruit et $a_n(t_0)<\epsilon_n$ pour $T$ suffisant.
\end{itemize}
Dans ces conditions, $C_T(t_0)<\epsilon_h(t_0)+\epsilon_b+\epsilon_n=\epsilon$ et $s(t_0)\approx b(t_0)$.
\end{proof}

La proposition~\ref{theo_prop} donne une méthode simple de sélection des points de $s(t)$ favorables à l'estimation de la ligne de base. L'estimation de la ligne de base se termine par une interpolation des points sélectionnés pour construire $\hat b(t)$.

\section{Simulation et résultats}
\label{sec_sr}

\paragraph{Performances de l'estimateur}

\begin{figure}
\centering
\includegraphics[width=1\columnwidth]{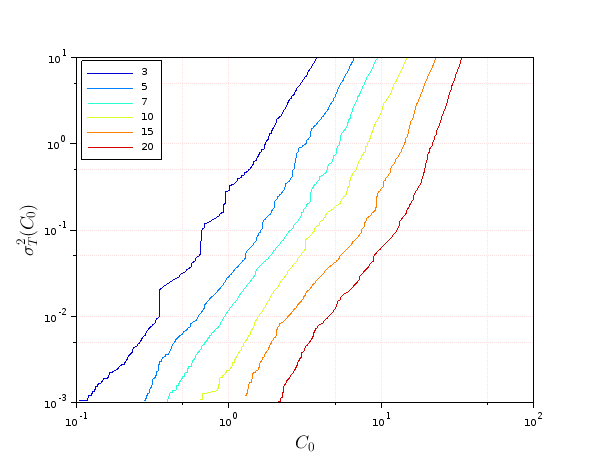}
\caption{EQM $\sigma^2_T(C_0)$ en fonction du seuil $C_0$ et de l'échelle d'observation $T\in\{3,5,7,10,15,20\}$.}
\label{fig_eqm}
\end{figure}

Le modèle présenté au paragraphe~\ref{sec_mm} est utilisé pour quantifier les performances de l'estimateur.  Nous commençons par valider la proposition~\ref{theo_prop} avec une EQM $\sigma^2_T(C_0)$ calculée uniquement sur les points $t_0$ \eqref{eq_t0} sélectionnés à l'échelle $T$ et pour un seuil $C_0$
\begin{align}
\sigma^2_T(C_0)=\frac{1}{|\mathcal{T}|}\sum_{t_0\in\mathcal{T}}|\hat b(t_0)-b(t_0)|^2\,,
\end{align}
avec
\begin{align}\label{eq_tco}
\mathcal{T}=\big\{t_0\ |\ C_T(t_0)<C_0\big\}
\end{align}
et $\hat b(t_0)=s(t_0)$. Ainsi, les points $t_0$ sont d'autant plus favorables que l'EQM $\sigma_T^2(C_0)$ est faible. La proposition~\ref{theo_prop} et \eqref{eq_tco} nous disent que $C_0$ doit également être faible.

La figure~\ref{fig_eqm} donne l'évolution de cette EQM en fonction du seuil $C_0$ pour plusieurs valeurs de l'échelle d'observation $T$. On observe bien une EQM qui décroît avec $C_0$. Ces résultats ne permettent cependant pas de sélectionner un seuil car 
\begin{itemize}
\item[1.] Les EQM, si elles sont faibles, sont des valeurs moyennes qui présentent de fortes dispersions. La proposition~\ref{theo_prop} n'est alors vérifiée qu'en moyenne ;
\item[2.] Le nombre de points favorables diminue avec $C_0$ et conduit à une impossibilité à suivre correctement la ligne de base, si ce nombre est trop faible ou  si la répartition des points n'est pas satisfaisante.
\end{itemize}
À noter que, du fait de l'échantillonnage irrégulier généré par la sélection des points $t_0$, une interpolation spline cubique conduit à des sur-oscillations limitant les performances de l'estimateur lorsque l'EQM est calculée sur l'ensemble des points de la ligne de base. L'interpolation linéaire est alors privilégiée.

\begin{table}
\centering
\begin{tabular}{@{}c|c@{$\quad$}cc@{$\quad$}ccccc@{}}
 &\multicolumn{2}{c}{$|\mathcal{T}|$}&\multicolumn{2}{c}{EQM}&\multicolumn{2}{@{}c@{}}{$P_\text{FA}$  ($10^{-3}$)}&\multicolumn{2}{@{}c@{}}{$P_\text{MD}$ ($10^{-3}$)}\\
 $T\,\big\backslash\,C_0$&0,1&1&0,1&1&0,1&1&0,1&1\\\hline
 3 &308&352&0,025&0,46 &1   &0,6&0,6&3,4\\
 5 &163&200&0,037&0,071&1,5 &0,9&0,7&1,3\\
 7 & 96&129&0,068&0,042&1,9 &1,3&0,9&0,9\\
 10& 47& 75&0,176&0,061&3   &1,8&1,6&0,8\\
 15& 17& 37&0,645&0,168&6,1 &2,9&3  &1,4\\
 20&  6& 19&2,02 &0,414&15,4&4,7&4,2&2,3\\\hline
\end{tabular}
\caption{Performances de l'estimateur.}
\label{tab_hatb}
\end{table}

Le tableau~\ref{tab_hatb} donne quelques valeurs moyennes du nombre de point sélectionnés, de l'EQM calculée sur tous les points de la ligne de base, des taux FA et MD pour plusieurs échelles d'observation et valeurs $C_0$. Contrairement aux résultats de la figure~\ref{fig_eqm}, il n'y a plus une simple décroissance de l'EQM, ou des $P_\text{FA}$ et $P_\text{MD}$, avec l'augmentation de l'échelle d'observation pour une valeur de $C_0$ donnée. Des valeurs de $T$ entre 5 et 10, voire 15, pour $C_0$ égale à 0,1 ou 1 donnent des EQM et des taux FA et MD comparables.

Les performances de la méthode proposée sont ensuite comparées à celles obtenues avec les méthodes BEADS ($d=1$, $f_c=0,03037$, $r=9,62$) \cite{DUVA15}, airPLS ($\lambda=125577$) \cite{ZHAN10} et de régression quantile \cite{KOMS11} développées pour estimer les lignes de bases sur des spectres de spectroscopie RMN ou chromatographie. Les valeurs données dans le tableau~\ref{tab_eqm} sont les EQM calculées sur 55 réponses $s(i)$ de 9858 échantillons, $i\in[1,9858]$. La puissance du bruit est $E[n(i)^2]=0{,}01$, $C_0=1$ et $T=5$. Selon cette EQM, la méthode proposée paraît bien adaptée pour traiter le problème de dérive de la ligne de base des capteurs d'humectation. Sur ces données simulées, elle est la plus performante devant BEADS, airPLS puis la régression polynomiale sur quantile, avec des écarts d'EQM allant de 0,009 à 1,73, cf. tableau~\ref{tab_eqm}.

\begin{table}
\centering
\begin{tabular}{@{}cccc@{}}
 BEADS & airPLS & régression & proposition\\\hline
 0,080 & 0,97 & 1,8 & 0,071\\\hline
\end{tabular}
\caption{EQM des estimateurs de ligne de base.}
\label{tab_eqm}
\end{table}

\paragraph{Mesures expérimentales}

Si la réduction du bruit est suffisante en moyennant les données sur 15 minutes, il reste des irrégularités locales à lisser. Une simple moyenne glissante est préférée au classique algorithme de Savitzky-Golay qui, en étalant les fronts abruptes, introduit des irrégularités supplémentaires en relation avec les coefficients négatifs du filtre et les fortes variations d'amplitude.  Un filtre de longueur 10 est utilisé. L'ensemble de la procédure d'estimation est alors la suivante
\begin{itemize}
\item[1.] Filtrage par moyenne glissante des données ;
\item[2.] Sélection des échantillons $\mathcal{T}$ \eqref{eq_tco};
\item[3.] Reconstruction de la ligne de base par interpolation linéaire des échantillons sélectionnés.
\end{itemize}

En prenant les configurations du tableau~\ref{tab_hatb} les plus favorables, on observe peu de différence sur les taux de FA, MD et sur la répartition $d(i)$, \eqref{eq_d}.  La figure~\ref{fig_famd} donne un exemple de taux de répartition de $d(i)$, pour $S=5$ sur un relevé de 8833 points de mesures, $C_0=1$ et $T=10$. L'utilisation de la méthode proposée, comparée à un seuil fixe, conduit à environ 10~\% de points de mesure humide en moins, soit 238 heures d'humectation en moins. 

\begin{figure}
\centering
\includegraphics[width=1\columnwidth]{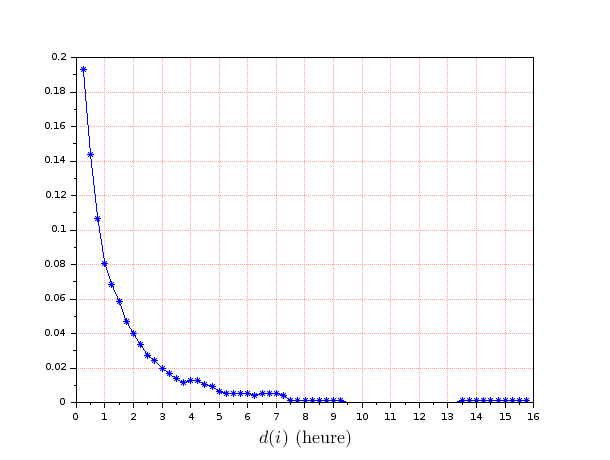}
\caption{Distribution de la distance à la phase humide $d(i)$.}
\label{fig_famd}
\end{figure}

\section{Conclusion}
\label{sec_conc}

La ligne de base des capteurs d'humectation foliaires
varie en environnement opérationnel et biaise la mesure des périodes sèches et humides. Dans ce travail, nous avons proposé une méthode d'estimation de cette ligne de base qui sélectionne, sur la base de la norme L1, des minimums locaux et une échelle d'observation. L'estimateur est validé sur un modèle de ligne de base et ses performances sont comparées à celle des algorithmes de la littérature. La mise en pratique de l'estimateur nécessite de fixer deux paramètres : un seuil et une échelle d'observation. L'analyse menée suggère une faible sensibilité de l'estimateur au choix de ces paramètres. De plus, la très faible complexité de l'algorithme et la faible profondeur de mémoire nécessaire permettent d'envisager des traitements \emph{in situ} en temps réel dans les microcontrôleurs intégrés dans les noeuds des capteurs déployés pour la surveillance des cultures. Toutefois, les performances de l'estimateur peuvent être dépendantes du modèle. Il est alors nécessaire de construire des bases de données expérimentales pour éprouver la robustesse de la proposition.

\shorthandoff:

\end{document}